\documentclass[10pt,twocolumn,aps,pra,superscriptaddress,floatfix]{revtex4-2}
\usepackage{amsmath,amssymb,amsfonts,amsthm,booktabs}
\usepackage{graphicx}
\usepackage[colorlinks=true,linkcolor=blue,citecolor=blue,urlcolor=blue]{hyperref}
\newcommand{\Cl}{\mathrm{Cliff}}
\newcommand{\F}{\mathbb F_2}
\newcommand{\Tr}{\operatorname{Tr}}
\newcommand{\Size}{\operatorname{Size}}
\newcommand{\supp}{\operatorname{supp}}
\newcommand{\E}{\mathbb E}

\newtheorem{theorem}{Theorem}
\newtheorem{lemma}[theorem]{Lemma}
\newtheorem{proposition}[theorem]{Proposition}
\newtheorem{corollary}[theorem]{Corollary}

\begin{document}

\title{Exact Clifford Optimality for Two-Sided Locally Randomized Classical Shadows}
\author{Xiaotian Nie}
\affiliation{Intelligent Quantum Inception Co., Ltd., Haidian, Beijing, 100083, China}
\affiliation{iFLYTEK Research, Hefei, 230088, China}
\affiliation{Beijing Key Laboratory of Cold Atom Quantum Computation, Tsinghua University, 100084, Beijing, China}
\author{Tao Zhang}
\email{tao\_zhang@mail.tsinghua.edu.cn}
\affiliation{Institute for Advanced Study, Tsinghua University, Beijing 100084, China}
\affiliation{Intelligent Quantum Inception Co., Ltd., Haidian, Beijing, 100083, China}
\affiliation{Beijing Key Laboratory of Cold Atom Quantum Computation, Tsinghua University, 100084, Beijing, China}
\author{Yadong Wu}
\email{yadongwu@scu.edu.cn}
\affiliation{College of Physics, Sichuan University, Chengdu 610065, China}
\author{Linghui Chen}
\email{lhchen@iflytek.com}
\affiliation{iFLYTEK Research, Hefei, 230088, China}
\affiliation{Intelligent Quantum Inception Co., Ltd., Haidian, Beijing, 100083, China}
\affiliation{Beijing Key Laboratory of Cold Atom Quantum Computation, Tsinghua University, 100084, Beijing, China}
\date{\today}

\begin{abstract}
Entangling measurements can reduce the cost of learning many-body
correlations. We consider classical shadows in which a fixed unitary
on a $k$-qubit region is sandwiched between two layers of random
single-qubit Clifford gates, and ask which middle unitary best
estimates Pauli strings with full support on the region.
Combining symplectic support counting with a constraint from
operator-size parity, we prove that within this architecture the
minimum squared shadow norm over all Clifford middle unitaries is
$C_k^{\rm opt}=4\cdot 9^k/(3\cdot 5^k+3^k-6)$ for odd $k$ and
$C_k^{\rm opt}=4\cdot 9^k/(3\cdot 5^k-3^k+2)$ for even $k$.
A dual-axis rotation circuit $U_*$ attains this minimum for every $k$.
Asymptotically $C_k^{\rm opt}\simeq(4/3)\cdot(9/5)^k$: the exponential
base of Wu \emph{et al.}'s contractive unitary is optimal, and the
prefactor improves from $2$ to $4/3$.
For $1\leq k\leq7$, exact certificates show that $U_*$ is
optimal over the entire unitary group $\mathrm U(2^k)$, and we
conjecture that this holds for every $k$.
\end{abstract}

\maketitle

\section{Introduction}
\label{sec:introduction}

Estimating many observables from limited quantum data is a basic task
in quantum simulation and information processing. Direct fidelity
estimation~\cite{FlammiaLiu2011,DaSilva2011DirectFidelity} and shadow
tomography~\cite{Aaronson2020ShadowTomography} illustrate how a selected
set of properties can be learned without reconstructing the entire
state. Classical shadows implement reusable property estimation through
randomized measurements and classical
post-processing~\cite{Huang2020}. The same measurement toolbox also
supports entanglement and correlation measurements
\cite{Elben2019Correlations,Brydges2019Renyi,Elben2023Toolbox}, and
shadow estimation has been demonstrated
experimentally~\cite{Struchalin2021Experimental,Zhang2021Experimental}.

For a Pauli string with $k$ nonidentity factors, independent random
single-qubit Pauli measurements give squared shadow norm $3^k$.
Random global Clifford measurements on the same $k$ qubits give
$2^k+1$. Entangling the measurement basis can therefore reduce the
statistical cost of long strings. In locally scrambled shadows, this
cost is set by the size distribution of the string after propagation through
the measurement circuit~\cite{Hu2023,Bu2024}, connecting measurement
design to entanglement features~\cite{YouGu2018,KuoAkhtarArovasYou2020}
and operator spreading~\cite{Nahum2018OperatorSpreading}.
Shallow random circuits, brickwork architectures, tensor-network
methods, and Hamiltonian evolution provide several routes to useful
measurement ensembles
\cite{Ippoliti2023OperatorRelaxation,Bertoni2024Shallow,
Arienzo2023Brickwork,Akhtar2023TensorNetworks,Hu2022Hamiltonian,Akhtar2025}.

\begin{figure}[!t]
\centering
\includegraphics[width=0.77\columnwidth]{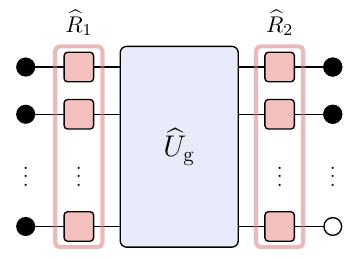}
\caption{Schematic of the classical shadow tomography protocol for a
Pauli string operator acting on a subsystem, following
Refs.~\cite{Hu2023,Wu2026}.
The protocol comprises a middle unitary $\widehat U_{\rm g}$ acting
globally on the subsystem (blue block), sandwiched between two layers $\widehat R_1$ and
$\widehat R_2$ of random single-qubit Clifford rotations
(pink blocks). Filled and open circles denote
nonidentity and identity factors of the propagated Pauli string,
respectively.}
\label{fig:measurement-protocol}
\end{figure}

A useful class of entangling measurement ensembles has the form
$\widehat R_2\widehat U_{\rm g}\widehat R_1$~\cite{Hu2023}
(Fig.~\ref{fig:measurement-protocol}): a fixed middle unitary
$\widehat U_{\rm g}$ on the selected $k$-qubit region is sandwiched
between two independent layers $\widehat R_1$ and $\widehat R_2$ of
uniformly random single-qubit Clifford gates.
The middle unitary is reused across shots, while the local gates are
resampled. The local layers make the measurement channel diagonal in
the Pauli basis, and a good middle unitary is one that shrinks
full-support strings before readout.

Wu \emph{et al.} constructed a contractive unitary~\cite{Wu2026}
that shrinks about half of all full-support strings to roughly
two-thirds of their length. They derived its exact squared shadow norm at
every $k$, which grows asymptotically as $2\cdot(9/5)^k$ and improves
the exponential base of global random Clifford measurements.
The exact optimum, however, was known only for two
qubits~\cite{Wu2026,Akhtar2025}.
This raises a central question: within this architecture, what is the
minimum squared shadow norm over all Clifford middle unitaries at each
block size $k$, and which circuit attains it?

We answer this question for every $k\geq1$. The minimum squared shadow
norm is
\begin{equation}
 C_k^{\rm opt}=
 \begin{cases}
 \dfrac{4\cdot 9^k}{3\cdot 5^k+3^k-6},&k\ \text{odd},\\[4pt]
 \dfrac{4\cdot 9^k}{3\cdot 5^k-3^k+2},&k\ \text{even}.
 \end{cases}
 \label{eq:intro-optimum}
\end{equation}
Optimality follows from counting how many full-support strings a
Clifford middle unitary can map into each output support. This
counting alone gives the base $9/5$; operator-size parity, which
links different output supports, sharpens it to the exact value at
every $k$.
The bound is attained by the dual-axis rotation circuit
$U_*=e^{i\pi X^{\otimes k}/4}e^{i\pi Z^{\otimes k}/4}$, composed of
two collective Pauli rotations. Compared with the contractive unitary,
it leaves only about a quarter of full-support strings at full size,
instead of half.
Asymptotically $C_k^{\rm opt}\simeq(4/3)\cdot(9/5)^k$, so both the
base and the prefactor are optimal, and the prefactor improves from
$2$ to $4/3$.
For $1\leq k\leq7$, exact certificates extend this
optimality to arbitrary middle unitaries, suggesting that the same
optimum holds over the full unitary group for every $k$.

The remainder of this paper is organized as follows.
Section~\ref{sec:protocol} defines the measurement protocol and states
the main results. Section~\ref{sec:bound} proves the Clifford bound
using binary symplectic linear algebra, quadratic forms, and
paired-coordinate selection related to linear matroid
parity~\cite{Lovasz1979Determinants,MatoyaOki2022Pfaffian,
IshikawaWakayama1995}, and extends it to arbitrary unitaries up to
seven qubits.
Section~\ref{sec:construction} gives the saturating construction
and compares it with the contractive unitary.
Section~\ref{sec:discussion} discusses
implications and open questions.
In the appendices, we provide supporting proofs and exact rational
certificates, and discuss reconstruction, implementation,
and application details.

\section{Measurement protocol and main results}
\label{sec:protocol}

\subsection{Physical setting and notation}

Let $\rho$ be a state of $N$ qubits, and let the subsystem $S$ be a
selected region of $k$ sites, adjacent or not, which we also call a
block of size $k$.
The main task is to estimate, from the same measurement
data, the expectation values of the family of Pauli observables that
are nonidentity at every site of $S$ and identity outside $S$.
We write $\widehat O$ for an arbitrary member of this family.
In single-block formulas, $\rho$ denotes the reduced state $\rho_S$
and $\widehat O$ the restricted Pauli $\widehat O_S$.

Following the two-sided locally randomized setting of
Refs.~\cite{Hu2023,Wu2026}, on each shot we sample the measurement circuit
\begin{equation}
 \widehat U=\widehat R_2\widehat U_{\rm g}\widehat R_1,\qquad
 \widehat R_j=\bigotimes_{i\in S}u_{j,i},
 \label{eq:ensemble}
\end{equation}
where all $u_{j,i}$ are independent and uniform single-qubit
Clifford gates. The middle unitary $\widehat U_{\rm g}$ is a fixed
Clifford acting globally on $S$.
After applying $\widehat U$, we measure in the computational basis,
obtaining $\mathbf z\in\{0,1\}^k$.
The raw snapshot and measurement channel are
\begin{align}
 \widehat\sigma_U(\mathbf z)
 &=\widehat U^\dagger|\mathbf z\rangle\langle\mathbf z|\widehat U,
 \label{eq:raw-snapshot}\\
 \mathcal M_{\mathcal E_U}(A)
 &=\E_{\widehat U\sim\mathcal E_U}\sum_{\mathbf z}
    \Tr[A\widehat\sigma_U(\mathbf z)]\,
             \widehat\sigma_U(\mathbf z).
 \label{eq:channel}
\end{align}
The canonical shadow is
$\widehat\rho=\mathcal M_{\mathcal E_U}^{-1}
(\widehat\sigma_U(\mathbf z))$, and the scalar single-shot estimate is
$o=\Tr(\widehat O\widehat\rho)$.

For a Pauli $P$, we write $\Size(P)=|\supp(P)|$ for the number of
nonidentity tensor factors.
Conjugating a uniformly selected full-support Pauli by
$\widehat U_{\rm g}$ gives an output-size distribution $\pi(m)$.
The output layer turns each nonidentity factor into $Z$ with
probability $1/3$, and the readout keeps only strings of $I$ and $Z$,
so a string of size $m$ survives with probability $3^{-m}$.
The two local layers make the channel diagonal in the Pauli basis,
with a common full-support eigenvalue
\begin{equation}
 \mathcal M_{\mathcal E_U}(\widehat O)=w_k\widehat O,\qquad
 w_k(\widehat U_{\rm g})=\sum_{m=0}^{k}\pi(m)3^{-m}.
 \label{eq:pauli-weight}
\end{equation}
Following Refs.~\cite{Wu2026,Hu2025RobustShallow}, we call the channel
eigenvalue $w_k$ the \emph{Pauli weight}.
This is distinct from the support size, which counts nonidentity
factors; when referring to support we use \emph{size} throughout.

\begin{proposition}[Statistical objective]
\label{prop:moment}
For every input state, the canonical estimator is unbiased.
Its state-independent second moment equals the squared shadow norm,
\begin{equation}
 C_k(\widehat U_{\rm g})
 :=\|\widehat O\|_{\mathcal E_U}^{\,2}
 =\E_\rho o^{\,2}
 =\frac1{w_k(\widehat U_{\rm g})}.
 \label{eq:norm}
\end{equation}
\end{proposition}

The proposition follows directly from Clifford Pauli propagation
(Appendix~\ref{app:estimator}). The squared shadow norm controls the
sampling cost: an average of $L$ independent single-shot estimates
has variance at most $C_k/L$.
For a fixed ensemble, $C_k$ is also the smallest worst-case variance
of any single-shot estimator unbiased for every state
(Corollary~\ref{cor:minimax}), so the optima below also hold when the
reconstruction is optimized.

\subsection{Main results}

For the proof, we rescale the Pauli weight into an integer score.
Let $\mathcal F_k=\{X,Y,Z\}^{\otimes k}$, with phases omitted, and define
\begin{equation}
 S_k(\widehat U_{\rm g})
 :=9^k w_k(\widehat U_{\rm g})
 =\sum_{P\in\mathcal F_k}
     3^{k-\Size(\widehat U_{\rm g}P\widehat U_{\rm g}^\dagger)}.
 \label{eq:score}
\end{equation}
\begin{theorem}[Exact Clifford optimum]
\label{thm:main}
For every $k\geq1$, set
\begin{equation}
 s_k=
 \begin{cases}
 (3\cdot 5^k+3^k-6)/4,&k\ \text{odd},\\
 (3\cdot 5^k-3^k+2)/4,&k\ \text{even}.
 \end{cases}
 \label{eq:sharp-score}
\end{equation}
In the ensemble \eqref{eq:ensemble},
\begin{equation}
 \max_{\widehat U_{\rm g}\in\Cl(k)}S_k(\widehat U_{\rm g})=s_k,
 \qquad
 \min_{\widehat U_{\rm g}\in\Cl(k)}C_k(\widehat U_{\rm g})
       =\frac{9^k}{s_k}.
 \label{eq:exact-optimum}
\end{equation}
The same optimum holds for mixtures of Clifford middle unitaries
drawn independently of the local layers.
\end{theorem}

\begin{corollary}[Optimal exponential base and prefactor]
\label{cor:asymptotic}
Let $C_k^{\rm opt}$ be the minimum in Eq.~\eqref{eq:exact-optimum}.
Then
\begin{equation}
 \lim_{k\to\infty}(C_k^{\rm opt})^{1/k}=\frac95,\qquad
 \lim_{k\to\infty}\frac{C_k^{\rm opt}}{(9/5)^k}=\frac43 .
 \label{eq:asymptotic}
\end{equation}
\end{corollary}
For odd $k\geq3$ the optimal normalized norm approaches $4/3$
from below; for even $k$ it approaches from above.

The same quantities are defined for any middle unitary
$\widehat U_{\rm g}\in\mathrm U(2^k)$: the local layers still make the
channel Pauli diagonal, and $\pi(m)$ becomes a squared-transfer
distribution (a general unitary maps each input Pauli to a
superposition; $\pi(m)$ organizes the squared expansion coefficients by
output size). Appendix~\ref{app:unitary-certificates} gives the
formulas and extends Proposition~\ref{prop:moment} to this setting.
For blocks of up to seven qubits, the Clifford optimum is also the
optimum over all middle unitaries.

\begin{theorem}[Finite-size unitary optimum]
\label{thm:finite-unitary}
At each block size $k=1,\ldots,7$,
\begin{equation}
 \max_{\widehat U_{\rm g}\in\mathrm U(2^k)}
       S_k(\widehat U_{\rm g})=s_k,\qquad
 \min_{\widehat U_{\rm g}\in\mathrm U(2^k)}
       C_k(\widehat U_{\rm g})=\frac{9^k}{s_k}.
 \label{eq:finite-unitary}
\end{equation}
The same optimum holds for mixtures of middle unitaries drawn
independently of the local layers.
\end{theorem}

Section~\ref{sec:bound} proves the bound $S_k\leq s_k$ in both
theorems, and Section~\ref{sec:construction} constructs a Clifford
circuit that attains it.

\section{Proof of the optimality bounds}
\label{sec:bound}

The proof has two steps. Counting how many full-support strings a
Clifford can map inside each output support gives
$S_k\leq5^k-2^k$, which already fixes the exponential base.
Operator-size parity then links the output supports: after an affine
shift, every full-support input lands on an output of parity
$k\bmod2$, and a bound on quadratic signs lowers the leading
coefficient from $1$ to the sharp $3/4$.
Section~\ref{sec:finite-unitary} replaces the counting by convex
relaxations to bound arbitrary unitaries.

\subsection{Symplectic labels and paired information sets}

Use binary Pauli labels $p\in V_k=\F^{2k}$~\cite{DehaeneDeMoor2003}, written as
$p=(x_1,z_1,\ldots,x_k,z_k)$.
The local labels $(1,0)$, $(0,1)$, and $(1,1)$ denote $X$, $Z$,
and $Y$, respectively. For $r=(x_1',z_1',\ldots,x_k',z_k')$, the
symplectic pairing
\begin{equation}
 B(p,r)=\sum_{i=1}^k(x_i z_i'+z_i x_i')
 \label{eq:symplectic}
\end{equation}
encodes commutation. Clifford conjugation acts through a symplectic
linear map $M$, so that Eq.~\eqref{eq:score} becomes
\begin{equation}
 S_k(M)=\sum_{p\in\mathcal F_k}g_k(Mp),
 \qquad g_k(y)=3^{k-\Size(y)}.
 \label{eq:binary-score}
\end{equation}
We use the same notation $\mathcal F_k$ for full-support binary labels.

For $T\subseteq[k]$, let $V_T$ consist of labels supported inside $T$.

\begin{lemma}[Paired information set]
\label{lem:paired}
Every nondegenerate $2\ell$-dimensional symplectic subspace
$L\subseteq V_k$ has a set of $\ell$ physical coordinate pairs
such that projection onto those pairs is bijective.
The same projection is an affine bijection on every translate of $L$.
In particular, a translate contains at most $3^\ell$ full-support labels.
\end{lemma}

Appendix~\ref{app:paired} proves the lemma by expanding a symplectic
Gram determinant. It is also a paired-coordinate information-set
statement for additive complementary-dual
codes over $\mathbb F_4$ with respect to the trace-Hermitian
dual~\cite{ShiLiuKimSole2023ACD}.

The first step expands $g_k$ over supports,
\begin{equation}
 g_k(y)=\sum_{\substack{T\subseteq[k]\\\supp(y)\subseteq T}}
                              2^{k-|T|}.
 \label{eq:support-expansion}
\end{equation}
The preimage $M^{-1}V_T$ is nondegenerate, so Lemma~\ref{lem:paired}
bounds its full-support intersection by $3^{|T|}$.
For $T=\varnothing$ the intersection is empty, and therefore
$S_k(M)\leq5^k-2^k$.
Size parity supplies the rest.

\subsection{Quadratic signs and affine duality}

The size parity is a quadratic function,
\begin{equation}
 Q_0(p)=\Size(p)\bmod2
       =\sum_{i=1}^k(x_i+z_i+x_i z_i),
 \label{eq:parity-quadratic}
\end{equation}
with polar form $B$ and Gauss sum
$\sum_p(-1)^{Q_0(p)}=(-2)^k$.

\begin{lemma}[Opposite-sign count]
\label{lem:quadratic}
Let $Q$ be a quadratic function on $\F^{2\ell}$, possibly with a
constant term, with an arbitrary nondegenerate polar form and Gauss
sum $\epsilon\,2^\ell$, where $\epsilon\in\{+1,-1\}$.
Then
\begin{equation}
 \#\{x\in\mathcal F_\ell:(-1)^{Q(x)}=-\epsilon\}
 \leq\left\lfloor\frac{3^{\ell+1}-3}{4}\right\rfloor .
 \label{eq:opposite-count}
\end{equation}
\end{lemma}

Appendix~\ref{app:quadratic} proves the lemma using a product measure
on physical coordinate pairs and the Fourier transform of a
nondegenerate quadratic phase.

For a symplectic map $H$, the function $Q_0(H^{-1}y)+Q_0(y)$ has zero
polar form, so it is linear and equals $B(\delta_H,y)$ for a unique
vector $\delta_H$:
\begin{equation}
 Q_0(H^{-1}y)+Q_0(y)=B(\delta_H,y).
 \label{eq:defect}
\end{equation}
Since $Q_0\circ H^{-1}$ has the same Gauss sum as $Q_0$, this forces
$Q_0(\delta_H)=0$, and hence
\begin{align}
 Q_0(Hp+\delta_H)&=Q_0(p),\label{eq:parity-forward}\\
 Q_0\bigl(H^{-1}(y+\delta_H)\bigr)&=Q_0(y).
 \label{eq:parity-inverse}
\end{align}
The auxiliary shift $\delta_H$ corrects the parity change induced by
$H$. All full-support inputs therefore have affine outputs with the
same parity $k\bmod2$, linking the output support sectors.
This allows the opposite-sign count to sharpen the separate support
bounds. The affine score is exactly dual to a physical
Clifford score:
\begin{equation}
 \sum_{p\in\mathcal F_k}g_k(Hp+\delta_H)=S_k(H^{-1}).
 \label{eq:affine-duality}
\end{equation}
This identity transfers bounds for the parity-preserving affine map
directly to the Clifford score (Appendix~\ref{app:duality}).

\subsection{Support counts and parity resummation}

Let
\begin{equation}
 N_T=\#\{p\in\mathcal F_k:Hp+\delta_H\in V_T\},
 \qquad \ell=|T|.
 \label{eq:affine-count}
\end{equation}
The affine space $H^{-1}(\delta_H+V_T)$ has nondegenerate direction
of dimension $2\ell$.
By Eq.~\eqref{eq:parity-inverse}, the Gauss sum of $Q_0$ on this
space is $(-2)^\ell$. Reparametrize it by the $\ell$ physical pairs
from Lemma~\ref{lem:paired}. The resulting quadratic function may
carry a constant term and a nonstandard polar form, but it stays
nondegenerate with the same Gauss sum.
Full-support labels project injectively into $\mathcal F_\ell$,
and their sign is $(-1)^k$.
When $k-\ell$ is odd, this sign is opposite to $(-1)^\ell$.
Lemma~\ref{lem:quadratic} therefore gives
\begin{equation}
 N_T\leq
 \begin{cases}
 (3^{\ell+1}-3)/4,&k\ \text{odd},\ \ell\ \text{even},\\
 (3^{\ell+1}-5)/4,&k\ \text{even},\ \ell\ \text{odd}.
 \end{cases}
 \label{eq:count-caps}
\end{equation}
The second line includes the integer rounding in
Eq.~\eqref{eq:opposite-count}.

For $p\in\mathcal F_k$, Eq.~\eqref{eq:parity-forward} says that
$r=k-\Size(Hp+\delta_H)$ is even.
For every label $y$ with this property,
\begin{equation}
 g_k(y)=1+
 2\!\!\sum_{\substack{T\supseteq\supp(y)\\k-|T|\ {\rm odd}}}
                              2^{k-|T|}.
 \label{eq:parity-resum}
\end{equation}
To verify the identity, sum over odd subsets of the $r$ zero
positions: their total weight is
$[(1+2)^r-(1-2)^r]/2=(3^r-1)/2$.
Equations~\eqref{eq:affine-duality} and \eqref{eq:parity-resum} yield
\begin{equation}
 S_k(H^{-1})
 =3^k+2\!\!\sum_{\substack{T\subseteq[k]\\k-|T|\ {\rm odd}}}
                         2^{k-|T|}N_T .
 \label{eq:score-resum}
\end{equation}
Only the support counts bounded in Eq.~\eqref{eq:count-caps} occur.

Substituting Eq.~\eqref{eq:count-caps} into
Eq.~\eqref{eq:score-resum} and evaluating the two parity-restricted
binomial sums gives $S_k(H^{-1})\leq s_k$ for both odd and even $k$
(Appendix~\ref{app:resummation}).
Inversion is a bijection on the symplectic group, so the bound
$S_k(M)\leq s_k$ holds for every Clifford $\widehat U_{\rm g}$.
It also holds for independent mixtures, since their channel
eigenvalues, and therefore their scores, average linearly.

\subsection{Choi-state certificates beyond Clifford}
\label{sec:finite-unitary}

For a general middle unitary, the bound comes from a convex
relaxation. Every unitary obeys linear constraints on the two-copy
swap statistics of its Choi state, and maximizing the score under
them bounds $S_k$ from above. For $k\leq6$ this linear program is
tight: Appendix~\ref{app:unitary-certificates} lists exact rational
dual certificates whose bounds are
\begin{equation}
 (s_1,s_2,s_3,s_4,s_5,s_6)=(3,17,99,449,2403,11537).
 \label{eq:finite-scores}
\end{equation}
At $k=7$ we use three copies of the Choi state instead. Positivity of
the three-copy moment and its partial transposes gives a semidefinite
program, and a computer-assisted dual certificate bounds it by
$s_7=59139$ (Appendix~\ref{app:seven-qubit}). Together these prove the
bound in Theorem~\ref{thm:finite-unitary}; mixtures are covered by the
same linearity. For $k=2$ the optimum also follows
from the transfer-parameter constraints and shadow-norm formula of
Ref.~\cite{Akhtar2025}.

\section{The dual-axis rotation circuit}
\label{sec:construction}

\subsection{Circuit and exact saturation}

The dual-axis rotation circuit is the Clifford circuit $U_*$ on a
$k$-qubit block defined by
\begin{equation}
 U_Z=e^{i\pi Z^{\otimes k}/4},\qquad
 U_X=e^{i\pi X^{\otimes k}/4},\qquad U_*=U_XU_Z .
 \label{eq:construction}
\end{equation}
\begin{proposition}[Attainment]
\label{prop:attainment}
The circuit $U_*$ attains the score $s_k$ for every $k\geq1$.
\end{proposition}

\begin{proof}
For a binary input label $p$, write
$a=\sum_i x_i$, $b=\sum_i z_i$, and $\eta=k\bmod2$.
The conjugation action of the two rotations is
\begin{equation}
 x_i'=x_i+b+\eta a,\qquad z_i'=z_i+a .
 \label{eq:binary-action}
\end{equation}
Thus every site receives the same shift
$c(v)=(b+\eta a,a)$, with $v=(a,b)$.
For $U_Z$ alone, the shift is $(0,a)$ and acts in two of the four
parity sectors. The second rotation makes $v\mapsto c(v)$ invertible,
so the shift is nonzero in all three sectors with $v\ne0$.
When $v=0$, no local label changes, so these inputs keep full size
and contribute
\begin{equation}
 A_k=\frac{3^k+3\cdot(-1)^k}{4}
 \label{eq:zero-sector}
\end{equation}
to the score.

In each nonzero sector, a site with $p_i=c(v)$ becomes an identity,
while the other two nonzero local labels remain nonzero.
The local score thus sums to $3+1+1=5$.
Resolving the two global parity constraints gives a leading score
$5^k/4$ from each of the three nonzero sectors, explaining the total
coefficient $3/4$.
Appendix~\ref{app:saturation} evaluates the exact sector sums,
including their parity-dependent corrections. Together with $A_k$,
they give Eq.~\eqref{eq:sharp-score}.
\end{proof}

Together with the bounds of Sec.~\ref{sec:bound},
Proposition~\ref{prop:attainment} proves Theorems~\ref{thm:main}
and~\ref{thm:finite-unitary}.

\subsection{Comparison with the contractive unitary}

The contractive unitary of Ref.~\cite{Wu2026} is
\begin{equation}
 \widehat U_{\rm ct}=\prod_{i<j}e^{i\pi Z_iZ_j/4}.
 \label{eq:wu-circuit}
\end{equation}
It is equivalent to $U_Z$ up to local Clifford gates and has score
\begin{equation}
 S_k(\widehat U_{\rm ct})
 =\frac{5^k+3^k+(-1)^k-1}{2}.
 \label{eq:wu-score}
\end{equation}
For a full-support input with $N_{XY}$ factors in $\{X,Y\}$,
the output size under $U_Z$ is $N_{XY}$ when $N_{XY}$ is odd,
and $k$ when it is even. Summing these cases gives
Eq.~\eqref{eq:wu-score}. The second rotation in $U_*$ uses all
three nonzero parity sectors to contract support, giving the
larger leading score $(3/4)\cdot 5^k$.

\begin{figure}[!t]
\centering
\includegraphics[width=\columnwidth]{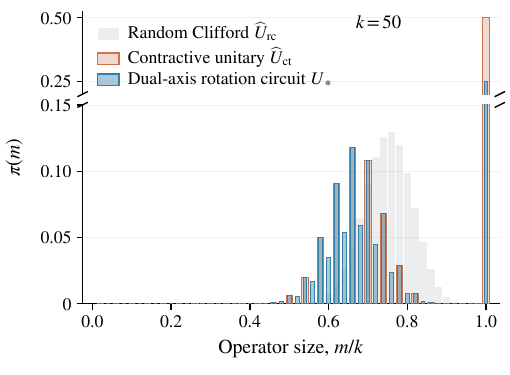}
\caption{Output-size distributions for a full-support input on
$k=50$ qubits. The distributions for random Clifford measurements and
the contractive unitary reproduce the comparison in Ref.~\cite{Wu2026}.
Below full size, $U_*$ matches $\widehat U_{\rm ct}$ at odd $m$ and
also fills even $m$.
The probability of retaining full support is approximately $1/2$ for
$\widehat U_{\rm ct}$ and $1/4$ for $U_*$. The vertical axis is broken
to enlarge the low-probability region.}
\label{fig:size-distribution}
\end{figure}

Figure~\ref{fig:size-distribution} compares the exact output-size
distributions at $k=50$. Both structured circuits have a contracted branch
near $m/k=2/3$, while $U_*$ lowers the probability remaining at
full size from approximately $1/2$ to $1/4$.
Because the readout weights size $m$ by $3^{-m}$, the contracted branch
dominates the weighted sum in Eq.~\eqref{eq:pauli-weight}, and for
$U_*$ its leading coefficient rises from $1/2$ to $3/4$.
Appendix~\ref{app:size-distributions} derives the exact distributions.

The score improvement is zero for $k=1,2$ and positive for all
$k\geq3$ (Appendix~\ref{app:score-comparison}). Even at $k=2$, where
both circuits reach the optimal score $17$, they differ: their
operator Schmidt ranks across the cut are $2$ and $4$, so they are
locally inequivalent (Appendix~\ref{app:implementation}).
Asymptotically,
\begin{equation}
 \frac{C_k(U_*)}{C_k(\widehat U_{\rm ct})}\longrightarrow\frac23.
 \label{eq:full-ratio}
\end{equation}
Figure~\ref{fig:performance} shows the exact finite-size comparison.

\begin{figure}[!htbp]
\centering
\includegraphics[width=\columnwidth]{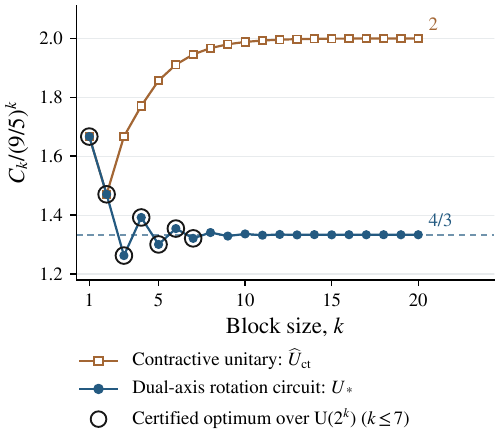}
\caption{Normalized squared shadow norm $C_k/(9/5)^k$ versus block
size $k$ for full-support Pauli observables. The dual-axis
rotation circuit $U_*$ attains the
Clifford optimum of Theorem~\ref{thm:main} for every $k$; dashed lines
show the asymptotic coefficients $2$ and $4/3$.
Black rings mark optimality over $\widehat U_{\rm g}\in\mathrm U(2^k)$
for $1\leq k\leq7$
(Theorem~\ref{thm:finite-unitary}).}
\label{fig:performance}
\end{figure}

With all-to-all connectivity, ancilla-free implementations use
$3k+O(1)$ CNOT gates for the dual-axis rotation circuit and
$2k+O(1)$ for the contractive unitary
(Appendix~\ref{app:implementation}).
The circuit $U_*$ thus trades gates for shots: in the noiseless
setting, its $2/3$ shot ratio in Eq.~\eqref{eq:full-ratio} offsets the
extra CNOTs to leading order, giving the same total CNOT cost at fixed
precision with one third fewer shots.

The channel eigenvalues obey a no-free-lunch constraint: for any
middle unitary, the eigenvalues of the $4^k-1$ nonidentity Paulis sum
to $2^k-1$ (Appendix~\ref{app:partial}). The circuit $U_*$ spends this
budget on strings that fill the block, so shorter strings pay: at
$k=3$, a single-site string has squared shadow norm $27$ under $U_*$
and $81/11$ under $\widehat U_{\rm ct}$.
For a contiguous length-$k$ string measured with uniformly shifted
blocks, as in Ref.~\cite{Wu2026}, $U_*$ still lowers the cost, with
asymptotic ratio $76/99$ (Appendix~\ref{app:sliding}).

\section{Discussion}
\label{sec:discussion}

The optimum has a simple structure. Symplectic support counting fixes
the base $9/5$, and operator-size parity, which links different output
supports, fixes the finite-size terms and, asymptotically, the
prefactor $4/3$.
The dual-axis rotation circuit reaches this bound at every block size
by contracting support in all three nonzero parity sectors.
Because the canonical estimator is minimax optimal for a fixed
ensemble (Corollary~\ref{cor:minimax}), the same optimum holds over
all single-shot estimators unbiased for every state.

Choi-state certificates show that $U_*$ remains optimal over the
entire unitary group $\mathrm U(2^k)$ for $1\leq k\leq7$.
We conjecture that $C_k^{\rm opt}=9^k/s_k$ is the optimum over all
unitaries for every $k$; a closed-form family of multi-copy Choi-state
certificates would be one route to a proof.

Outside this architecture, the base $9/5$ can be beaten:
locally randomized GHZ-basis measurements reach $3/2$ for
full-support strings~\cite{Ippoliti2024,ChenGongYe2024}, but only
when the string fills the GHZ block exactly. Once a single site is
missing, this advantage fails suddenly and totally, leaving the string
impossible to estimate or as costly as under single-qubit Pauli
measurements. The two-sided architecture is informationally complete,
and $U_*$ keeps the base $9/5$ when a fixed number of sites is missing
(Appendix~\ref{app:partial}), so its advantage extends beyond exactly
full support.

Within the two-sided architecture, the exact optimum gives a benchmark
for any Clifford circuit. Turning it into an experimental advantage
raises two open questions: the minimum CNOT count needed to reach it,
with all-to-all connectivity and on a nearest-neighbor chain, and the
best balance between CNOT count and shot count under two-qubit gate
noise.

\section*{Author contributions}

X.N. conceived the key idea and developed the initial proofs.
X.N., T.Z., and Y.W. reviewed the proofs. All authors discussed the
results and contributed to writing the manuscript.

\section*{Use of AI tools}

OpenAI's GPT-6 Astra assisted in reformulating the authors' proofs
with full rigor, checking the arguments, and polishing the manuscript. Anthropic's Claude Opus 5.5 assisted
in polishing the manuscript. The authors reviewed and refined all
resulting material and take full responsibility for the manuscript
and its results.

\appendix

\section{Estimator identities and reconstruction}
\label{app:estimator}

\subsection{Channel eigenvalues and the second moment}

For a fixed sampled Clifford $\widehat U$, write
$\widehat U\widehat O\widehat U^\dagger=\pm P'$.
Computational-basis dephasing keeps this Pauli precisely when
$P'$ consists only of $I$ and $Z$ factors.
Averaging the output local layer makes this event have probability
$3^{-m}$ conditional on size $m$. The input local layer maps a
full-support target uniformly over the $3^k$ phase-free
full-support Paulis. This proves Eqs.~\eqref{eq:pauli-weight}
and \eqref{eq:score}.
The same argument applies to every support set, so all nonidentity
Pauli eigenvalues are positive and the measurement channel is
invertible.

Since the channel is self-adjoint and Pauli diagonal,
\begin{equation}
 o=\frac1{w_k}
 \langle\mathbf z|\widehat U\widehat O
                         \widehat U^\dagger|\mathbf z\rangle .
 \label{eq:scalar-estimator}
\end{equation}
If the propagated Pauli is not diagonal, this quantity is zero.
If it is diagonal, its numerator is $\pm1$ for every outcome.
Its squared numerator is therefore the indicator of a circuit
event independent of the input state and measurement outcome.
Averaging gives $\E_\rho o^{\,2}=w_k/w_k^2=1/w_k$.
Unbiasedness follows from the inverse-channel definition.
Taking the supremum over states gives the squared shadow norm,
proving Proposition~\ref{prop:moment}.
The same proof applies to independent Clifford mixtures and to
the averaged channel of the sliding protocol.

For a mixture, the canonical inverse uses the eigenvalue averaged
over the mixture.

For the explicit circuits in this article, closed-form eigenvalues
give the exact reconstruction normalization.
Once that normalization is known, stabilizer propagation computes
the numerator of Eq.~\eqref{eq:scalar-estimator}
efficiently~\cite{AaronsonGottesman2004}.

\subsection{Minimax single-shot reconstruction}
\label{app:minimax}

The fixed-ensemble reconstruction problem can be formulated using
measurement frames~\cite{DArianoPerinotti2007,Innocenti2023Frames}. Here Pauli symmetry
lets the canonical inverse attain the lower bound on the maximum
variance over input states.

\begin{corollary}[Minimax reconstruction]
\label{cor:minimax}
Fix an ensemble of the form \eqref{eq:ensemble}, allowing an arbitrary
middle unitary or an independent mixture of such unitaries.
For a full-support Pauli $\widehat O$ on $k\geq1$ qubits, let
$\mathcal X_{\widehat O}$ be the real single-shot estimators based on
the sampled circuit and outcome that are unbiased for every state.
Then
\begin{equation}
 \inf_{x\in\mathcal X_{\widehat O}}
      \sup_\rho\operatorname{Var}_\rho(x)=\frac1{w_k},
 \label{eq:minimax-variance}
\end{equation}
and the canonical estimator attains this value.
\end{corollary}

\begin{proof}
Set $d=2^k$, $E_{U,\mathbf z}=\widehat\sigma_U(\mathbf z)$, and
$a(U,\mathbf z)=\Tr(\widehat O E_{U,\mathbf z})$.
Unbiasedness for all states is equivalent to the operator identity
\begin{equation}
 \E_{\widehat U}\sum_{\mathbf z}
       x(U,\mathbf z)E_{U,\mathbf z}=\widehat O.
 \label{eq:unbiased-frame-identity}
\end{equation}
At $\rho_0=I/d$, every outcome has conditional probability $1/d$.
Writing $\E_0$ for expectation under this state, we obtain
\begin{align}
 \E_0[xa]&=\frac{\Tr(\widehat O^2)}d=1,\nonumber\\
 \E_0[a^2]&=
 \frac{\Tr[\widehat O\mathcal M_{\mathcal E_U}(\widehat O)]}{d}
 =w_k.
 \label{eq:minimax-cauchy-data}
\end{align}
Cauchy--Schwarz gives $\E_0[x^2]\geq1/w_k$.
Since $\widehat O$ is traceless, the estimator has zero mean at
$\rho_0$, so its worst-case variance is at least $1/w_k$.
The canonical estimator $x=a/w_k$ has second moment $1/w_k$
for every state, by Proposition~\ref{prop:moment} in the Clifford
case and Eq.~\eqref{eq:general-moment} for arbitrary middle unitaries.
Its variance is $1/w_k-\Tr(\rho\widehat O)^2$, with supremum $1/w_k$
attained at $\rho_0$. This saturates the lower bound.
The bound also covers estimators with extra classical randomness:
averaging over the seed preserves unbiasedness and never raises the
variance for any state.
\end{proof}

Consequently, the exact optima in Theorems~\ref{thm:main} and
\ref{thm:finite-unitary} also apply when reconstruction is optimized
over this class of estimators.

\section{Proof of the paired information-set lemma}
\label{app:paired}

Let $G$ be a $2\ell\times2k$ row-basis matrix of the subspace $L$,
and let $J$ be block diagonal with $k$ blocks
$\left(\begin{smallmatrix}0&1\\1&0\end{smallmatrix}\right)$.
Nondegeneracy of $L$ means
\begin{equation}
 \det(GJG^{\mathsf T})=1
\end{equation}
over $\F$. Let $\sigma$ exchange the two columns within every
physical coordinate pair.
For a set $I$ of $2\ell$ column indices, write $G_I$ for the
corresponding square submatrix.
Cauchy--Binet gives, with signs immaterial in characteristic two,
\begin{equation}
 1=\sum_{|I|=2\ell}\det(G_I)\det(G_{\sigma(I)}).
\end{equation}
If $I\ne\sigma(I)$, its term equals that for $\sigma(I)$, so they
cancel. Each remaining $I$ is a union of $\ell$ complete physical
pairs. At least one has $\det(G_I)=1$.
Projection onto those pairs is therefore an isomorphism on $L$.

On a translate $a+L$, the same projection is an affine bijection.
Every full-support point projects to a nonzero local label on
each selected pair. There are only $3^\ell$ such projected
labels, proving the count in Lemma~\ref{lem:paired}.

\section{Proof of the quadratic-sign bound}
\label{app:quadratic}

We prove Lemma~\ref{lem:quadratic} on $\F^{2\ell}$.
For a quadratic function with a possible constant term, its polar
form is
\[
 B_Q(x,y)=Q(x+y)+Q(x)+Q(y)+Q(0).
\]
Adding a constant to $Q$ changes its phase and the sign $\epsilon$
together, leaving the counted set unchanged.
We may therefore normalize $Q(0)=0$, keeping the argument fixed,
since a translation would move $\mathcal F_\ell$.

\subsection{A product-measure bound}

On each physical pair define the distribution
\begin{equation}
 \mu(0)=\frac12,\qquad \mu(x)=\frac16\quad(x\ne0).
 \label{eq:measure}
\end{equation}
Equivalently,
$\mu=\frac13\delta_0+\frac23\operatorname{Unif}(\F^2)$.
The product $\mu^{\otimes\ell}$ is a mixture of uniform measures
on coordinate subspaces.

For any quadratic function $R$ with $R(0)=0$, its phase average
over a linear subspace is either zero or a signed power
$2^{-r/2}$, where $r$ is the even rank of its restricted polar
form. To see this, if the restricted function is nonzero on the
polar radical, summing over that radical gives zero.
Otherwise it descends to a nondegenerate quadratic function
on a quotient, whose symplectic pairs each contribute a factor
of magnitude $1/2$.
For $r=0$ the average is zero or one; for $r\geq2$ it is at
least $-1/2$.
The empty coordinate subspace has average one and occurs in the
mixture with probability $3^{-\ell}$.
Consequently
\begin{equation}
 \E_{\mu^{\otimes\ell}}(-1)^R
 \geq-\frac12+\frac32\cdot 3^{-\ell}.
 \label{eq:product-bound}
\end{equation}

\subsection{Fourier transform of the quadratic phase}

Use the unnormalized Fourier transform defined by the standard
physical symplectic pairing:
\begin{equation}
 \widehat f(a)=\sum_{x\in\F^{2\ell}}f(x)(-1)^{B(a,x)} .
 \label{eq:fourier}
\end{equation}
Set $f_\ell=\mathbf1_{\mathcal F_\ell}$ and
$g_\ell(a)=3^{\ell-\Size(a)}$.
A one-pair calculation gives
\begin{equation}
 \widehat f_\ell(a)=(-1)^{Q_0(a)}g_\ell(a).
 \label{eq:fourier-full}
\end{equation}
Here $Q_0$ is the size-parity function on $\ell$ pairs.

Let $B_Q$ be the polar form of the normalized $Q$.
It may differ from $B$. Its nondegeneracy gives a unique
linear map $t:\F^{2\ell}\to\F^{2\ell}$ specified by
\begin{equation}
 B_Q(t(a),x)=B(a,x)\quad\text{for all }x.
\end{equation}
Completing the square,
$Q(x)+B(a,x)=Q(x+t(a))+Q(t(a))$.
Thus, with $Q^*(a)=Q(t(a))$,
\begin{equation}
 \widehat{(-1)^Q}(a)
 =\epsilon\,2^\ell\cdot(-1)^{Q^*(a)},\qquad Q^*(0)=0 .
 \label{eq:fourier-quadratic}
\end{equation}

Let $Z=\sum_{x\in\mathcal F_\ell}(-1)^{Q(x)}$.
Parseval's identity gives
\begin{align}
 \epsilon Z
 &=2^{-\ell}\sum_a
       g_\ell(a)(-1)^{Q^*(a)+Q_0(a)}\nonumber\\
 &=3^\ell\,\E_{\mu^{\otimes\ell}}
                       (-1)^{Q^*+Q_0}
 \geq\frac{3-3^\ell}{2}.
 \label{eq:quadratic-parseval}
\end{align}
The second equality follows from $\sum_a g_\ell(a)=6^\ell$
and Eq.~\eqref{eq:measure}.
The function $Q^*+Q_0$ vanishes at zero, so
Eq.~\eqref{eq:product-bound} applies even if its polar form is
degenerate.
The required count is $(3^\ell-\epsilon Z)/2$.
Taking its integer part proves Eq.~\eqref{eq:opposite-count},
including $\ell=0$.

\section{Affine duality and parity resummation}
\label{app:duality}

\subsection{Defect identities and the inverse-map duality}

Since $H$ is symplectic, the function
$Q_0(H^{-1}y)+Q_0(y)$ has zero polar form and is linear.
Nondegeneracy of $B$ gives the unique $\delta_H$ in
Eq.~\eqref{eq:defect}.
For any $d$,
\begin{equation}
 \sum_y(-1)^{Q_0(y)+B(d,y)}
 =(-1)^{Q_0(d)}\cdot(-2)^k .
\end{equation}
The invertibility of $H$ leaves the Gauss sum of
$Q_0\circ H^{-1}$ unchanged. Substituting $d=\delta_H$
therefore forces $Q_0(\delta_H)=0$.
Expanding $Q_0(Hp+\delta_H)$ and
$Q_0(H^{-1}(y+\delta_H))$ then gives
Eqs.~\eqref{eq:parity-forward} and \eqref{eq:parity-inverse}.

The Fourier transform of $g_k$ is
\begin{equation}
 \widehat g_k(a)=2^k g_k(a).
 \label{eq:fourier-g}
\end{equation}
Together with Eq.~\eqref{eq:fourier-full}, it implies for every
symplectic $K$ that
\begin{equation}
 S_k(K)=2^{-k}\sum_a(-1)^{Q_0(a)}g_k(a)g_k(Ka).
 \label{eq:linear-fourier-score}
\end{equation}
Define the auxiliary affine score
\[
 T_{\delta_H}(H)=\sum_{p\in\mathcal F_k}g_k(Hp+\delta_H).
\]
Fourier inversion and
$B(a,Hp)=B(H^{-1}a,p)$ give
\begin{align}
 T_{\delta_H}(H)
 &=2^{-k}\sum_a g_k(a)
   (-1)^{B(\delta_H,a)+Q_0(H^{-1}a)}
                       g_k(H^{-1}a)\nonumber\\
 &=2^{-k}\sum_a(-1)^{Q_0(a)}
                       g_k(a)g_k(H^{-1}a)\nonumber\\
 &=S_k(H^{-1}).
 \label{eq:duality-proof}
\end{align}
The second step uses Eq.~\eqref{eq:defect}, and the last uses
Eq.~\eqref{eq:linear-fourier-score} with $K=H^{-1}$.

\subsection{Completion of the sharp bound}
\label{app:resummation}

Substituting the support-count bounds in Eq.~\eqref{eq:count-caps}
into Eq.~\eqref{eq:score-resum} requires two binomial sums:
\begin{align}
 \sum_{k-|T|\ {\rm odd}}2^{k-|T|}\cdot 3^{|T|}
       &=\frac{5^k-1}{2},\label{eq:binomial-a}\\
 \sum_{k-|T|\ {\rm odd}}2^{k-|T|}
       &=\frac{3^k-(-1)^k}{2}.
 \label{eq:binomial-d}
\end{align}
For odd $k$, substitution gives
\begin{equation}
 S_k(H^{-1})
 \leq3^k+\frac34\cdot(5^k-1)-\frac34\cdot(3^k+1)
 =\frac{3\cdot 5^k+3^k-6}{4}.
\end{equation}
For even $k$, it gives
\begin{equation}
 S_k(H^{-1})
 \leq3^k+\frac34\cdot(5^k-1)-\frac54\cdot(3^k-1)
 =\frac{3\cdot 5^k-3^k+2}{4}.
\end{equation}
Both bounds equal $s_k$ in Eq.~\eqref{eq:sharp-score}.

\section{Exact finite-size certificates over the full unitary group}
\label{app:unitary-certificates}

\subsection{General Pauli transfer and the statistical objective}

Let $\widehat U_{\rm g}$ now be an arbitrary $k$-qubit unitary
and $d=2^k$. For phase-free Paulis
$P,Q\in\mathcal P_k=\{I,X,Y,Z\}^{\otimes k}$, define
\begin{equation}
 \mathcal R_{\rm g}(Q,P)
 =\frac1d\Tr\!\left(
 Q\widehat U_{\rm g}P\widehat U_{\rm g}^{\dagger}\right).
 \label{eq:general-transfer}
\end{equation}
These real coefficients form an orthogonal Pauli transfer matrix.
The two local layers keep the measurement channel Pauli diagonal.
Its full-support eigenvalue is
\begin{equation}
 w_k=\frac1{3^k}\sum_{P\in\mathcal F_k}
                 \sum_{Q\in\mathcal P_k}
       3^{-\Size(Q)}\mathcal R_{\rm g}(Q,P)^2 .
 \label{eq:general-weight}
\end{equation}
Equivalently, the size distribution in Eq.~\eqref{eq:pauli-weight}
is the squared-transfer distribution
\[
 \pi(m)=\frac1{3^k}\sum_{P\in\mathcal F_k}
                   \sum_{\Size(Q)=m}\mathcal R_{\rm g}(Q,P)^2 .
\]
It is nonnegative and normalized. Analogous expressions hold on
every input support. Orthogonality and $3^{-\Size(Q)}\geq3^{-k}$
make all Pauli eigenvalues positive, so the full channel is invertible.
The score extension is
\begin{equation}
 S_k(\widehat U_{\rm g})
 =\sum_{P\in\mathcal F_k}\sum_{Q\in\mathcal P_k}
       3^{k-\Size(Q)}\mathcal R_{\rm g}(Q,P)^2 .
 \label{eq:general-score}
\end{equation}
For a Clifford $\widehat U_{\rm g}$, each input column has a single
nonzero squared entry, recovering Eq.~\eqref{eq:score}.

The second-moment identity also survives this extension, as follows
from Pauli invariance (Proposition~4 of Ref.~\cite{Bu2024}).
Write $E_{U,\mathbf z}=\widehat\sigma_U(\mathbf z)$.
Self-adjointness of the Pauli-diagonal inverse channel gives
$o=w_k^{-1}\Tr(\widehat O E_{U,\mathbf z})$.
Right multiplication $\widehat U\mapsto\widehat U A$ by any
Pauli $A$ leaves the ensemble invariant. It changes
$\Tr(\widehat O E_{U,\mathbf z})$ only by a sign.
Averaging the Born factor over $A$ gives
$\Tr(E_{U,\mathbf z})/d=1/d$ for every input state. Therefore
\begin{align}
 \E_\rho o^{\,2}
 &=\frac1{d w_k^2}\E_{\widehat U}\sum_{\mathbf z}
                  [\Tr(\widehat O E_{U,\mathbf z})]^2\nonumber\\
 &=\frac{\Tr[\widehat O\mathcal M_{\mathcal E_U}
                              (\widehat O)]}{d w_k^2}
 =\frac1{w_k}.
 \label{eq:general-moment}
\end{align}

\begin{table*}[t]
\caption{Exact rational certificates for optimality over the entire unitary
group $\mathrm U(2^k)$ at $1\leq k\leq6$. Left: row and column multipliers. Right: nonnegative
Pauli-positivity multipliers. Entries not listed are zero; in particular,
all $\mu_{ab}$ vanish for $k\leq3$. The sign convention is fixed by
Eq.~\eqref{eq:certificate-residual}.}
\label{tab:certificate-multipliers}
\centering
\begin{minipage}[t]{0.53\textwidth}
\vspace{0pt}
\centering
\begin{tabular}{rrrr}
\toprule
$k$ & $r$ & $y_r$ & $z_r$\\
\midrule
1 & 0 & $1$ & $0$\\
\addlinespace[2pt]
2 & 0 & $0$ & $1$\\
2 & 2 & $8$ & $0$\\
\addlinespace[2pt]
3 & 0 & $0$ & $1$\\
3 & 2 & $8$ & $0$\\
\addlinespace[2pt]
4 & 0 & $0$ & $13/4$\\
4 & 2 & $7/2$ & $-7/4$\\
4 & 3 & $3/2$ & $0$\\
4 & 4 & $71$ & $9/4$\\
\addlinespace[2pt]
5 & 0 & $1831/256$ & $0$\\
5 & 1 & $0$ & $-165/256$\\
5 & 2 & $2519/256$ & $-337/32$\\
5 & 3 & $15/16$ & $915/256$\\
5 & 4 & $15751/256$ & $-241/16$\\
5 & 5 & $15/8$ & $1995/256$\\
\addlinespace[2pt]
6 & 0 & $0$ & $3901/256$\\
6 & 1 & $-285/256$ & $0$\\
6 & 2 & $1017/512$ & $-2855/512$\\
6 & 3 & $-245/512$ & $2963/512$\\
6 & 4 & $5965/256$ & $-199/32$\\
6 & 5 & $315/64$ & $335/256$\\
6 & 6 & $575011/512$ & $2655/512$\\
\bottomrule
\end{tabular}
\end{minipage}\hfill
\begin{minipage}[t]{0.43\textwidth}
\vspace{0pt}
\centering
\begin{tabular}{rrr}
\toprule
$k$ & $(a,b)$ & $\mu_{ab}$\\
\midrule
4 & $(1,1)$ & $3/32$\\
4 & $(1,2)$ & $1/32$\\
\addlinespace[2pt]
5 & $(1,1)$ & $1/8$\\
5 & $(1,2)$ & $5/256$\\
5 & $(1,3)$ & $1/256$\\
5 & $(3,1)$ & $7/256$\\
5 & $(4,1)$ & $5/256$\\
\addlinespace[2pt]
6 & $(1,1)$ & $35/256$\\
6 & $(1,2)$ & $33/1024$\\
6 & $(1,3)$ & $17/4096$\\
6 & $(2,1)$ & $3/1024$\\
6 & $(2,4)$ & $1/4096$\\
6 & $(3,1)$ & $7/4096$\\
6 & $(4,1)$ & $21/512$\\
6 & $(4,2)$ & $31/4096$\\
\bottomrule
\end{tabular}
\end{minipage}
\end{table*}

\subsection{Necessary constraints from the Choi state}

Use the normalized Choi vector
\[
 |u\rangle=(I\otimes\widehat U_{\rm g})
                    d^{-1/2}\sum_{j=0}^{d-1}|j,j\rangle .
\]
On two copies, consider a swap test on each input-qubit pair and
each output-qubit pair.
Let $p_{rs}$ be the probability of $r$ antisymmetric input
outcomes and $s$ antisymmetric output outcomes, and set
$W_{rs}=4^k p_{rs}$.
These probabilities obey
\begin{align}
 W_{rs}&\geq0,\qquad
 W_{rs}=0\quad\text{if }r+s\text{ is odd},\label{eq:choi-parity}\\
 \sum_sW_{rs}&=\nu_r,\qquad \sum_rW_{rs}=\nu_s,\qquad
 \nu_r=\binom{k}{r}\cdot 3^{k-r}.
 \label{eq:choi-marginals}
\end{align}
The parity constraint follows because two copies of a pure state
have positive global swap parity.
Each Choi marginal is maximally mixed, giving the binomial row
and column sums.

Define coefficients $L_a(r)$ by
\begin{equation}
 \sum_{a=0}^kL_a(r)x^a=(1+x)^{k-r}(1-3x)^r.
 \label{eq:choi-polynomials}
\end{equation}
The identity
$\sum_{\sigma=X,Y,Z}\sigma\otimes\sigma=2F-I$,
where $F$ swaps two local copies, gives
\begin{align}
 A_{ab}
 &:=\sum_{\substack{\Size(P)=a\\\Size(Q)=b}}
                          \mathcal R_{\rm g}(Q,P)^2\nonumber\\
 &=4^{-k}\sum_{r,s}L_a(r)L_b(s)W_{rs}\geq0.
 \label{eq:choi-positivity}
\end{align}
This relation applies the correspondence between swap-test statistics
and quantum weight enumerators~\cite{Shi2026Swap,Miller2026Enumerators}
to the Choi state, resolving input and output sizes separately.
The input transpose in the Choi correspondence causes no sign
change here: the sign of $Y$ cancels between two copies.
Since $L_k(r)=(-3)^r$ and
$\sum_b3^{k-b}L_b(s)=4^k\delta_{s0}$, Eq.~\eqref{eq:general-score}
becomes
\begin{equation}
 S_k(\widehat U_{\rm g})=\sum_{r\ {\rm even}}3^rW_{r0}.
 \label{eq:choi-objective}
\end{equation}
Every unitary satisfies
Eqs.~\eqref{eq:choi-parity}--\eqref{eq:choi-positivity}, so
maximizing the score over all $W$ obeying them bounds $S_k$ from
above.

\subsection{Explicit rational upper certificates}

Choose rational $y_r,z_s$ and nonnegative rational $\mu_{ab}$.
For every cell with $r+s$ even, define
\begin{equation}
 h_{rs}=3^r\delta_{s0}-y_r-z_s+
                \sum_{a,b}\mu_{ab}L_a(r)L_b(s).
 \label{eq:certificate-residual}
\end{equation}
If every $h_{rs}\leq0$, multiplication by $W_{rs}\geq0$
and summation yield
\begin{align}
 S_k(\widehat U_{\rm g})
 &\leq\sum_r(y_r+z_r)\nu_r
              -4^k\sum_{a,b}\mu_{ab}A_{ab}\nonumber\\
 &\leq\sum_r(y_r+z_r)\nu_r.
 \label{eq:certificate-bound}
\end{align}
Table~\ref{tab:certificate-multipliers} lists all nonzero
multipliers for $k=1,\ldots,6$; omitted entries are zero.
Exact rational evaluation gives
$\max_{r+s\ {\rm even}}h_{rs}=0$ in every case, and the last
sum in Eq.~\eqref{eq:certificate-bound} equals $s_k$.
Proposition~\ref{prop:attainment} supplies $S_k(U_*)=s_k$,
proving Theorem~\ref{thm:finite-unitary} for $k\leq6$.

\subsection{Three-copy certificate at seven qubits}
\label{app:seven-qubit}

At $k=7$ the two-copy linear program has maximum $s_7+432/31$, so we
pass to three copies of the Choi state. After averaging over the
local Clifford layers, the third moment
$\E\,(|u\rangle\langle u|)^{\otimes3}$ reduces to finitely many real
coordinates $x$, and the score is linear in them,
$S_7=-5^7\,q^{\mathsf T}x$. Every unitary yields an $x$ obeying three
families of constraints: $715$ linear equations $Ex=e$ from the
maximally mixed Choi marginals; scalar inequalities $Ix\leq l$, of
which the certificate uses $31$ nonnegative squared-transfer sums and
$20$ lower bounds on reduced purities; and positive semidefiniteness
of the moment blocks $F_s(x)$ and of their partial transposes
$F_p(x)$ (transposing one of the three Choi copies). A dual point
$(y,\mu\geq0,Z_s\succeq0,Z_p\succeq0)$ satisfying
\begin{equation}
 q+E^{\mathsf T}y+I^{\mathsf T}\mu
 -\sum_sF_s^*(Z_s)-\sum_pF_p^*(Z_p)=0
 \label{eq:three-copy-dual}
\end{equation}
then gives $S_7(\widehat U_{\rm g})\leq5^7(e^{\mathsf T}y+l^{\mathsf T}\mu)$
for every unitary.

The dual is defined by exact rational linear systems, and
complementarity with the exact rational moment of $U_*$ fixes its
value at $s_7=59139$. Ball arithmetic at $192$ bits encloses the
solution and proves positivity of all $340$ moment blocks and $384$
partial-transpose blocks: every Schur-complement pivot has a strictly
positive certified lower bound.
A numerical semidefinite solver was used only to find the rational
starting point, and the verification uses no optimizer.
Together with Proposition~\ref{prop:attainment}, this proves
Theorem~\ref{thm:finite-unitary} at $k=7$.

\section{Circuit comparison, size distributions, and implementation}
\label{app:resources}

\subsection{Exact saturation by parity sectors}
\label{app:saturation}

We evaluate the sector contributions to the score of $U_*$ in
Proposition~\ref{prop:attainment}. Write
$v=(a,b)=\sum_i p_i$ and $c(v)=(b+\eta a,a)$, with $\eta=k\bmod2$,
as in Eq.~\eqref{eq:binary-action}.
The zero sector contributes $A_k$ in Eq.~\eqref{eq:zero-sector}.

For $v\ne0$, let $c=c(v)$.
At one site the weighted character sum is
\begin{equation}
 K_c(t)=\sum_{p_i\ne0}g_1(p_i+c)(-1)^{B(t,p_i)}
 =\begin{cases}
 5,&t=0,\\
 1,&t=c,\\
 -3,&t\notin\{0,c\}.
 \end{cases}
 \label{eq:local-character}
\end{equation}
Projecting onto the input constraint $\sum_i p_i=v$ gives
\begin{align}
 \frac14\sum_{t\in\F^2}(-1)^{B(t,v)}K_c(t)^k
 &=\frac{5^k+s(v)-(-3)^k[1+s(v)]}{4},
 \label{eq:sector-score}\\
 s(v)&=(-1)^{B(c(v),v)}.\nonumber
\end{align}
For odd $k$, the three nonzero $v$ all have $s(v)=-1$.
For even $k$, two have $s(v)=-1$ and one has $s(v)=+1$.
Adding these sectors to Eq.~\eqref{eq:zero-sector} gives precisely
Eq.~\eqref{eq:sharp-score}.

\subsection{Comparison with the contractive unitary}
\label{app:score-comparison}

With $a=\sum_i x_i$ and $\eta=k\bmod2$, the binary action
of $\widehat U_{\rm ct}$ is
\begin{equation}
 x_i'=x_i,\qquad z_i'=z_i+a+\eta x_i .
\end{equation}
For $U_Z$ the last term $\eta x_i$ is absent.
That difference is a product of local phase gates.
Any residual difference in the Clifford action's signs is also
a Pauli correction, hence local.
Local Clifford factors on either side of the middle unitary
are absorbed by the uniform layers of Eq.~\eqref{eq:ensemble}.
This proves the equivalence of their channel spectra.

Subtracting Eq.~\eqref{eq:wu-score} from Eq.~\eqref{eq:sharp-score}
gives the score improvement
\begin{equation}
 s_k-S_k(\widehat U_{\rm ct})=
 \begin{cases}
 (5^k-3^k-2)/4,&k\ \text{odd},\\
 (5^k-3\cdot 3^k+2)/4,&k\ \text{even}.
 \end{cases}
 \label{eq:gain}
\end{equation}
It is zero for $k=1,2$ and positive for all $k\geq3$.

\subsection{Full-support size distributions}
\label{app:size-distributions}

The output-size distribution can be encoded by
\begin{equation}
 W_{\widehat U_{\rm g}}(t)=
 \sum_{P\in\mathcal F_k}
 t^{\Size(\widehat U_{\rm g}P\widehat U_{\rm g}^\dagger)},\qquad
 S_k=3^k W_{\widehat U_{\rm g}}(1/3).
 \label{eq:size-enumerator}
\end{equation}
Pauli weight enumerators also organize quantum-code data
\cite{ShorLaflamme1997,Rains1998Weight}; the family here is the
full-support input set, and the channel fixes the evaluation point.

The enumerator in Eq.~\eqref{eq:size-enumerator} determines the
probabilities by $\pi(m)=3^{-k}[t^m]W_{\widehat U_{\rm g}}(t)$.
For the contractive unitary, separating odd and even $N_{XY}$ gives
\begin{align}
 W_{\widehat U_{\rm ct}}(t)
 &=\frac{3^k+(-1)^k}{2}\,t^k\nonumber\\
 &\quad+\frac{(1+2t)^k-(1-2t)^k}{2}.
 \label{eq:ct-size-enumerator}
\end{align}
For the dual-axis rotation circuit,
\begin{equation}
 W_{U_*}(t)=A_k t^k+D_k(t),
 \label{eq:star-size-enumerator}
\end{equation}
where $A_k$ is given in Eq.~\eqref{eq:zero-sector}, and
\begin{equation}
 D_k(t)=
 \begin{cases}
 \dfrac{3\cdot[(1+2t)^k-(1-2t)^k]}4,&k\ \text{odd},\\[5pt]
 \dfrac{3\cdot(1+2t)^k-(1-2t)^k-2}4,&k\ \text{even}.
 \end{cases}
 \label{eq:contracted-size-enumerator}
\end{equation}
To obtain this expression, the zero-parity sector contributes
$A_k t^k$. In each nonzero sector, the local size-generating
character takes the values $1+2t$, $1-2t$, and $-1$ at Fourier
labels $0$, $c$, and the remaining two labels, respectively.
The character projection used in Eq.~\eqref{eq:sector-score}, with
the same signs $s(v)$, then gives Eq.~\eqref{eq:contracted-size-enumerator}.

For comparison, a uniformly random $k$-qubit Clifford maps a
nonidentity input uniformly over the $4^k-1$ nonidentity Paulis.
Its output-size probabilities are therefore
\begin{equation}
 \pi_{\rm rc}(m)=\frac{\binom{k}{m}\cdot 3^m}{4^k-1},
 \qquad 1\leq m\leq k.
 \label{eq:random-size-distribution}
\end{equation}

\subsection{Implementation and local inequivalence}
\label{app:implementation}

Each Pauli rotation can be synthesized by accumulating parity,
applying a one-qubit phase rotation, and uncomputing the parity.
The usual construction uses $2(k-1)$ CNOT gates for a $k$-site
Pauli string~\cite{Amy2019CNOTPhase,CowtanEtAl2020}.
For the angle in Eq.~\eqref{eq:construction}, this count reduces
to $2k-3$ when $k\geq2$.
Choose an edge $(a,b)$ of a spanning tree and accumulate the parity
of each of the two components onto its endpoint, using $k-2$ CNOTs.
The remaining two-site rotation satisfies
\begin{equation}
 e^{i\pi Z_aZ_b/4}=e^{i\pi/4}S_a^\dagger S_b^\dagger\,\mathrm{CZ}_{ab},
 \qquad S=\operatorname{diag}(1,i).
 \label{eq:two-site-rotation}
\end{equation}
Since CZ is locally equivalent to one CNOT, applying this rotation
and reversing the parity accumulation uses $2(k-2)+1=2k-3$ CNOTs.
Hadamard basis changes give the same count for $U_X$.
No ancillas are needed. The same bound applies to the contractive
unitary, which is locally Clifford equivalent to $U_Z$.

The two rotations in the dual-axis rotation circuit can be compiled
jointly~\cite{CowtanEtAl2020}. Let $m=\lfloor k/2\rfloor$ and define
\begin{equation}
 W=\prod_{j=1}^{m}\mathrm{CNOT}_{2j-1\to2j}=W^\dagger.
 \label{eq:pairing-circuit}
\end{equation}
Let $A$ contain the odd-indexed sites and $B$ the even-indexed sites,
with the unpaired site $k$ also included in $B$ when $k$ is odd.
Thus $|A|=|B|=r:=\lceil k/2\rceil$.
Writing $X_A=\prod_{i\in A}X_i$ and $Z_B=\prod_{i\in B}Z_i$, we have
\begin{equation}
 WX^{\otimes k}W^\dagger=X_A,\qquad
 WZ^{\otimes k}W^\dagger=Z_B.
 \label{eq:paired-supports}
\end{equation}
Consequently,
\begin{equation}
 U_*=W e^{i\pi X_A/4}e^{i\pi Z_B/4}W^\dagger.
 \label{eq:joint-compilation}
\end{equation}
The pairing and its inverse use $2m$ CNOT gates.
Each remaining rotation acts on $r$ sites and uses $2r-3$ CNOTs
for $r\geq2$, or none for $r=1$.
The total is therefore $2m+2(2r-3)$ for $k\geq3$, giving
the implementation bounds
\begin{equation}
 N_{\mathrm{CNOT}}(U_*)\leq
 \begin{cases}
 0,&k=1,\\
 2,&k=2,\\
 3k-6,&k\geq4\ \text{even},\\
 3k-5,&k\geq3\ \text{odd}.
 \end{cases}
 \label{eq:joint-cnot-count}
\end{equation}
These counts assume access to the pairing edges and connected
parity-accumulation networks within $A$ and $B$, as available with
all-to-all connectivity.
On a nearest-neighbor chain, compiling the two rotations separately
uses $4k-6$ CNOTs for $k\geq2$.

For any nonempty bipartition of the block at $k\geq2$,
\begin{equation}
 U_Z=\frac{I+iZ^{\otimes k}}{\sqrt2}
\end{equation}
has operator Schmidt rank two.
The expansion
\begin{equation}
 U_*=\frac12\cdot\left[
 I+iZ^{\otimes k}+iX^{\otimes k}
             -X^{\otimes k}Z^{\otimes k}\right]
\end{equation}
has rank four: on each side of the cut the four factors are
distinct Pauli strings and therefore orthogonal.
Multiplication by local unitaries on the input and output
preserves this rank. Thus $U_*$ and $U_Z$ are locally
inequivalent, including $k=2$ where both are optimal.

\section{Exact partial-support channel eigenvalues}
\label{app:partial}

Let $\tilde k$ be the number of active sites in a $k$-qubit block,
denote the corresponding Pauli by $\widehat O_{\tilde k,k}$,
and put $q=k-\tilde k$.
For a fixed block size $k$, the channel eigenvalue depends on the
target support only through its cardinality $\tilde k$, because
$U_*$ is permutation invariant.
Uniform input randomization averages the $3^{\tilde k}$
phase-free Paulis on this support.
For either permutation-invariant construction,
\begin{equation}
 w_{\tilde k,k}
 =\frac{S_{\tilde k,k}}{3^{k+\tilde k}},\qquad
 S_{\tilde k,k}
 =\sum_{\substack{p_i\ne0\ {\rm on\ the\ support}\\
                  p_i=0\ {\rm elsewhere}}}
                           3^{k-\Size(Mp)} .
 \label{eq:partial-definition}
\end{equation}

Define
\begin{equation}
 A_{\tilde k}=\frac{3^{\tilde k}+3\cdot(-1)^{\tilde k}}4 .
\end{equation}
For $U_*$, the exact score is
\begin{align}
 S_{\tilde k,k}^{*}
 &=3^{k-\tilde k}A_{\tilde k}+D_{\tilde k,k},\label{eq:partial-star}\\
 D_{\tilde k,k}
 &=
 \begin{cases}
 3\cdot(5^{\tilde k}-1)/4,& k\ \text{odd},\\
 [3\cdot 5^{\tilde k}-2\cdot(-3)^{\tilde k}-1]/4,& k\ \text{even}.
 \end{cases}
 \label{eq:partial-d}
\end{align}
To derive the formula, use the parity sectors of
Eq.~\eqref{eq:binary-action}, summing only over active input
sites. The zero-shift sector contributes $A_{\tilde k}$ choices
and leaves all $q$ inactive sites as identities, supplying the
factor $3^q$.
Every nonzero-shift sector turns all inactive sites into
nonidentity labels; their weight factor is one.
Equation~\eqref{eq:sector-score}, with the number of active
sites in its powers replaced by $\tilde k$ and the parity map
$c(v)$ still determined by the block size $k$, then gives
$D_{\tilde k,k}$.

For the contractive unitary,
\begin{equation}
 S_{\tilde k,k}^{\rm ct}
 =\frac{3^{k-\tilde k}\cdot[3^{\tilde k}+(-1)^{\tilde k}]
               +5^{\tilde k}-1}{2}.
 \label{eq:partial-wu-score}
\end{equation}
If $N_{XY}$ among the active letters is even, the output size
is $\tilde k$. If it is odd, the active $Z$ letters are erased
and the $q$ inactive sites become nonidentity, giving size
$q+N_{XY}$. Summing these two cases proves the formula.
Dividing by $3^{k+\tilde k}$ gives the form used in
Ref.~\cite{Wu2026},
\begin{align}
 w_{\tilde k,k}^{\rm ct}
 &=\frac12\cdot\left[3^{-\tilde k}
                      +(-1)^{\tilde k}\cdot 9^{-\tilde k}\right]\nonumber\\
 &\quad+\frac12\cdot\left[(5/9)^{\tilde k}
                             -9^{-\tilde k}\right]\cdot 3^{-(k-\tilde k)}.
 \label{eq:partial-wu}
\end{align}
At $\tilde k=0$, both scores are $3^k$ and both eigenvalues
are one. At $\tilde k=k$, they reduce to the full-support
scores in Eqs.~\eqref{eq:sharp-score} and \eqref{eq:wu-score}.
For a fixed number $q$ of missing sites, the term $D_{\tilde k,k}$
dominates Eq.~\eqref{eq:partial-star} at large $k$, giving
\begin{equation}
 \|\widehat O_{k-q,k}\|_{\mathcal E_*}^2
 \sim\frac43\cdot 3^q\cdot(9/5)^{k-q}.
 \label{eq:partial-asymptotic}
\end{equation}
The base $9/5$ is unchanged: relative to full support at the same
block size, each missing site raises the cost by a factor of only
$5/3$.

The two circuits redistribute a fixed total. Each rank-one term
$A\mapsto\Tr[A\widehat\sigma_U(\mathbf z)]\,\widehat\sigma_U(\mathbf z)$
of Eq.~\eqref{eq:channel} has trace one, so the measurement channel
has trace $2^k$ for every middle unitary. The eigenvalues of the
$4^k-1$ nonidentity Paulis on a block therefore always sum to $2^k-1$,
and a larger full-support eigenvalue is paid for on other supports.
Relative to the contractive unitary, $U_*$ pays on short strings;
for example, at $k=3$ and $\tilde k=1$,
\begin{equation}
 \|\widehat O_{1,3}\|_{\mathcal E_*}^2=27,\qquad
 \|\widehat O_{1,3}\|_{\mathcal E_{\rm ct}}^2=\frac{81}{11}.
 \label{eq:partial-counterexample}
\end{equation}

\section{Shifted measurement blocks and uniform asymptotics}
\label{app:sliding}

\subsection{Protocol and statistical cost}

We apply the shifted-block framework of Ref.~\cite{Wu2026},
replacing the block unitary $\widehat U_{\rm ct}$ by $U_*$.
Consider an infinite chain, or a periodic chain of $N=n_0k$ sites
with $n_0\geq2$, and a target with $k$ consecutive nonidentity
factors. The measurement now acts on the whole chain: on each shot,
select uniformly one of the $k$ offsets of a partition into $k$-site
blocks, and measure every block independently with
$\widehat R_2\widehat U_{\rm g}\widehat R_1$
(Fig.~\ref{fig:sliding-blocks}).
One offset aligns the target with a block; the others split it
into $k_1$ and $k-k_1$ active sites in two distinct blocks.
For either circuit, the offset-averaged eigenvalue is
\begin{equation}
 w_{\rm slide}=\frac1k\left[
 w_{k,k}+\sum_{k_1=1}^{k-1}
                    w_{k_1,k}w_{k-k_1,k}\right].
 \label{eq:sliding-eigenvalue}
\end{equation}
The tensor-product channel factorizes across blocks even for
entangled input states. Canonical reconstruction inverts this
averaged eigenvalue, so $C_{\rm slide}=1/w_{\rm slide}$.

\begin{proposition}[Sliding-block cost]
\label{prop:sliding}
For the geometry and estimator just specified,
\begin{align}
 C_{\rm slide}(U_*)&\sim\frac{128}{99}\,k(9/5)^k,\label{eq:sliding-new}\\
 C_{\rm slide}(\widehat U_{\rm ct})&\sim
                   \frac{32}{19}\,k(9/5)^k.\label{eq:sliding-old}
\end{align}
\end{proposition}
The proof below gives error bounds uniform in the offset.
The asymptotic cost ratio is $76/99$,
different from the full-block ratio $2/3$ because split targets
probe the partial-support spectrum.
Figure~\ref{fig:sliding-performance} shows the exact finite-size comparison.

\begin{figure*}[t]
\centering
\includegraphics[width=0.94\textwidth]{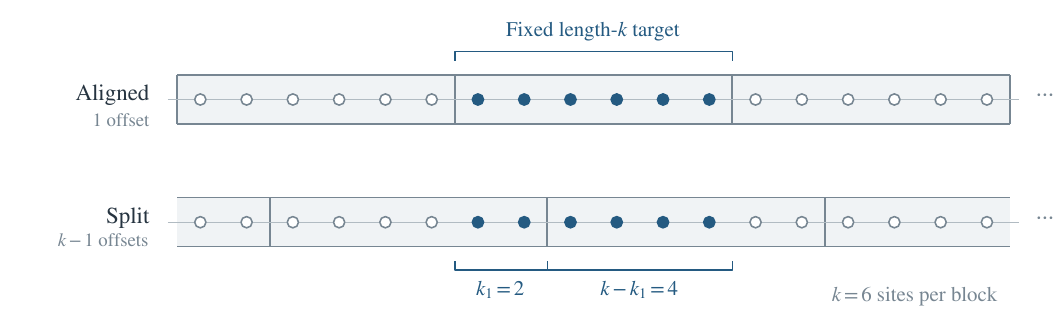}
\caption{Uniformly shifted blocks, illustrated at $k=6$.
The target occupies the same six filled sites in both rows.
Each measurement block contains $k$ sites; open sites carry identity
factors of the target observable. One offset aligns the target
with a block. The other $k-1$ offsets split it between neighboring
blocks; the example has $k_1=2$ and $k-k_1=4$.}
\label{fig:sliding-blocks}
\end{figure*}

\begin{figure}[t]
\centering
\includegraphics[width=\columnwidth]{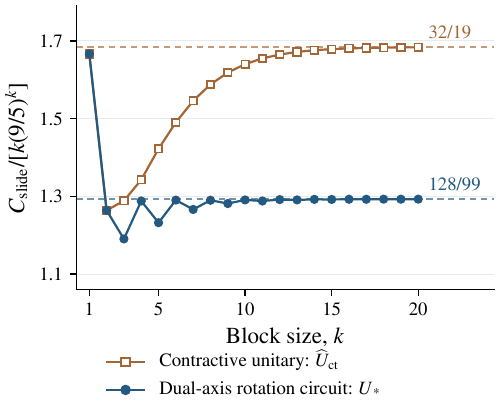}
\caption{Normalized squared shadow norm $C_{\rm slide}/[k(9/5)^k]$
versus block size $k$ for a length-$k$ contiguous target measured with
uniformly shifted $k$-site blocks. The curves compare the contractive unitary and the
dual-axis rotation circuit,
with asymptotic coefficients $32/19$ and $128/99$.}
\label{fig:sliding-performance}
\end{figure}

\subsection{Uniform asymptotics}

For active size $0\leq j\leq k$, write
\begin{align}
 w_{j,k}^{*}&=a_j+d_k r^j+e_{k,j},\label{eq:sliding-decomp}\\
 a_j&=\frac{3^{-j}+3\cdot(-1/9)^j}{4},\quad
 d_k=\frac34\cdot 3^{-k},\quad r=\frac53,\nonumber\\
 e_{k,j}
 &=\begin{cases}
 -\frac34\cdot 3^{-k-j},&k\ \text{odd},\\
 -\frac12\cdot(-1)^j\cdot 3^{-k}-\frac14\cdot 3^{-k-j},&k\ \text{even}.
 \end{cases}
 \label{eq:sliding-remainder}
\end{align}
These identities follow exactly from
Eqs.~\eqref{eq:partial-star}--\eqref{eq:partial-d}.
Uniformly in $j$, $|a_j|\leq3^{-j}$ and
$|e_{k,j}|\leq3^{-k}$.
Define $D_k=d_k r^k=(3/4)\cdot(5/9)^k$.
The aligned contribution is
$w_{k,k}^{*}=D_k+O(3^{-k})$.

In the convolution of Eq.~\eqref{eq:sliding-eigenvalue},
the two $a$--$d$ terms give
\begin{equation}
 2D_k\sum_{j=1}^{k-1}a_jr^{-j}.
\end{equation}
The other terms have the following uniform bounds:
\begin{align}
 \sum_{j=1}^{k-1}|a_j a_{k-j}|&=O(k3^{-k}),\nonumber\\
 \sum_{j=1}^{k-1}d_k^2r^k&=O\bigl(k(5/27)^k\bigr),\nonumber\\
 \sum_{j=1}^{k-1}|a_j e_{k,k-j}|&=O(3^{-k}),\label{eq:uniform-errors}\\
 \sum_{j=1}^{k-1}d_k r^j|e_{k,k-j}|
                                      &=O\bigl((5/27)^k\bigr),\nonumber\\
 \sum_{j=1}^{k-1}|e_{k,j}e_{k,k-j}|&=O(k9^{-k}).\nonumber
\end{align}
The symmetric mixed terms satisfy the same bounds.
Each is $o(D_k)$, despite the increasing number of offsets.
Moreover,
\begin{equation}
 \sum_{j=1}^{\infty}a_jr^{-j}
 =\frac14\sum_{j=1}^{\infty}[5^{-j}+3\cdot(-1/15)^j]
 =\frac1{64}.
\end{equation}
It follows that
\begin{equation}
 k w_{\rm slide}^{*}
 =D_k\left(1+\frac2{64}+o(1)\right)
 =\frac{99}{128}\cdot(5/9)^k[1+o(1)].
 \label{eq:sliding-limit}
\end{equation}

For the contractive unitary the exact decomposition is
\begin{equation}
 w_{j,k}^{\rm ct}
 =\frac{3^{-j}+(-1/9)^j}{2}
       +\frac12\cdot 3^{-k}r^j-\frac12\cdot 3^{-k-j}.
\end{equation}
The same uniform bounds apply.
The aligned term has leading coefficient $1/2$, and the
cross-term series is
\begin{equation}
 \frac12\sum_{j=1}^{\infty}[5^{-j}+(-1/15)^j]=\frac3{32}.
\end{equation}
Hence $k w_{\rm slide}^{\rm ct}\sim(19/32)\cdot(5/9)^k$.
Taking the inverse proves Proposition~\ref{prop:sliding}.

\end{document}